\documentclass[lettersize,journal]{IEEEtran}

\usepackage{amsmath,amssymb}
\usepackage{subfigure}
\usepackage{hyperref}
\usepackage{graphicx,graphics,color,psfrag}
\usepackage{cite,balance}
\usepackage{caption}
\allowdisplaybreaks
\usepackage{algorithm}
\usepackage{algorithmic}
\usepackage{accents}
\usepackage{amsthm}
\usepackage{bm}
\usepackage{url}
\usepackage[english]{babel}
\usepackage{multirow}
\usepackage{enumerate}
\usepackage{cases}
\usepackage{stfloats}
\usepackage{dsfont}
\usepackage{color,soul}
\usepackage{amsfonts}
\usepackage{cite,graphicx,amsmath,amssymb}
\usepackage{subfigure}
\usepackage{fancyhdr}
\usepackage{hhline}
\usepackage{graphicx,graphics}
\usepackage{array,color}
\usepackage{mathtools}
\usepackage{amsmath}
\usepackage{amsthm}
\usepackage{caption}
\usepackage{circledsteps}

\include{header}
\newtheorem{theorem}{Theorem}

\newtheorem{remark}{Remark}

\usepackage{algorithm}
\def\BibTeX{{\rm B\kern-.05em{\sc i\kern-.025em b}\kern-.08em
		T\kern-.1667em\lower.7ex\hbox{E}\kern-.125emX}}
\begin{document}

	\title{Movable Antenna Enabled Symbiotic Radio Systems: An Opportunity for Mutualism}
	
		\author{\IEEEauthorblockN{Chao Zhou,~\IEEEmembership{Student Member,~IEEE},~Bin Lyu,~\IEEEmembership{Senior Member,~IEEE},~Changsheng You,~\IEEEmembership{Member,~IEEE},\\and~Ziwei Liu}

		\IEEEcompsocitemizethanks{\IEEEcompsocthanksitem C. Zhou,  B. Lyu, and Z. Liu are with the {School of Communications and Information Engineering}, Nanjing University of Posts and Telecommunications, Nanjing 210003, China (email: zoe961992059@163.com, blyu@njupt.edu.cn, liuziwei89@163.com).  
		C. You is with the Department of Electrical and Electronic
		Engineering, Southern University of Science and Technology (SUSTech),
		Shenzhen 518055, China (e-mail: youcs@sustech.edu.cn).
		%Y.-C. Liang is with the Center for Intelligent Networking and Communications (CINC), University of Electronic Science and Technology of China (UESTC), Chengdu 611731, China (e-mail: liangyc@ieee.org).
		%D. T. Hoang is with School of Electrical and Data Engineering, University of Technology Sydney, Sydney, NSW 2007, Australia (email: hoang.dinh@uts.edu.au). 	
		}
		%\IEEEauthorblockA{\IEEEauthorrefmark{1} Nanjing University of Posts and Telecommunications, Nanjing 210003, China}
		%\IEEEauthorblockA{\IEEEauthorrefmark{2}  University of Technology Sydney, Sydney, NSW 2007, Australia}
		% \IEEEauthorblockA{\IEEEauthorrefmark{3}  Sun Yat-sen University, Guangzhou 510275,  China}
	}
	
	\markboth{}%
	{Shell \MakeLowercase{\textit{et al.}}: A Sample Article Using IEEEtran.cls for IEEE Journals}

	\maketitle
	
	\begin{abstract}

	In this letter, we propose a new movable antenna (MA) enabled symbiotic radio (SR) system that leverages the movement of MAs to maximize both the primary and secondary rates, thereby promoting their mutualism. Specifically, the primary transmitter (PT) equipped with MAs utilizes a maximum ratio transmission (MRT) beamforming scheme to ensure the highest primary rate at the primary user (PU). Concurrently, the backscatter device (BD) establishes the secondary transmission by overlaying onto the primary signal. The utilization of MAs aims to enhance the secondary rate by optimizing the positions of MAs to improve the beam gain at the BD. Accordingly, the beam gains for both MA and fixed-position antenna (FPA) scenarios are analyzed, confirming the effectiveness of the MA scheme in achieving the highest primary and secondary rates. Numerical results verity the superiority of our proposed MA enabled scheme.

	\end{abstract}
	

	\section{Introduction}

	\IEEEPARstart{T}{he} evolution of wireless communication networks has considerably contributed to the proliferation of the Internet of Things (IoT), and it is expected that the quantity of wireless IoT devices will escalate to 5000 billion by 2030~\cite{IoT_Num}. Along with the increasing automation and intelligence of IoT networks, there are two urgent issues that need to be addressed. On one hand, the data transmissions of vast IoT devices require a significant amount of power consumption. On the other hand, the ubiquitous connectivity of IoT devices will result in substantial occupation of spectrum resources, thereby exacerbating the problem of spectrum scarcity. The challenges posed by energy consumption and spectrum occupancy inevitably limit the future development of IoT networks, thus underscoring the need for enhancing energy-efficiency and spectrum utilization in IoT networks.
		
	Fortunately, a passive IoT communication paradigm, known as \emph{symbiotic radio} (SR), has been proposed to support energy- and spectrum-efficient IoT transmission~\cite{SR_TCCN,SR_Li}. Specifically, the IoT transmission in SR systems, also referred to as secondary transmission, is achieved by overlaying onto the primary transmission, thereby effectively reusing the primary radio frequency (RF) source and spectrum band. Driven by this superior advantage, a majority of research has focused on advancing SR systems to achieve highly effective primary and secondary transmissions~\cite{SR_BD,SR_Zeng,Multi-Acc_Liang,PSR_Yang,PSR_Zeng}. 
	In~\cite{SR_BD}, two practical setups, known as commensal SR (CSR) and parasitic SR (PSR), were introduced. Under these setups, the maximization of weighted sum rate and the minimization of transmit power were investigated. In~\cite{SR_Zeng}, the deployment of massive backscatter devices (BDs) was proposed to operate in the CSR setup, aiming to improve the performance of primary transmission by establishing additional reflection links.  Furthermore, under the CSR setup, two multiple access schemes with a large number of IoT devices were designed in~\cite{Multi-Acc_Liang} to ensure the reliability and validity of primary transmission. In~\cite{PSR_Yang}, the millimeter wave SR system was studied under the PSR setup, revealing the trade-off between primary and secondary transmissions. Moreover, the cell-free enabled SR system was proposed in~\cite{PSR_Zeng}, where the achievable rate-region of primary and secondary transmissions was investigated, confirming the competitive relationship between primary and secondary transmissions in the PSR setup.
	However, in the CSR setup, prioritizing the enhancement of primary transmission performance may affect the secondary transmission performance. Additionally, in the PSR setup, the presence of secondary transmission has a negative impact on the primary transmission performance. As a result, there exists a trade-off between primary and secondary transmissions~\cite{SR_BD}, while existing works fail to meet both primary and secondary requirements. Thus, how to guarantee the quality of service (QoS) requirements of both primary and secondary transmissions is still an open problem.

	Recently, \emph{movable antennas} (MAs) have received growing research attention due to its superior exploitation of spatial diversity~\citen{MA_mag,You_TransceiverDesign}. In~\cite{MA_Capacity}, the potential of MAs in enhancing channel capacity through optimizing the positions of antennas was demonstrated. Additionally, MAs provide an advanced capability to enhance the desired signal power and suppress undesired interferences, presenting new opportunities for secure transmission~\cite{Secure_MA}.
	 Furthermore, the movement of MAs can improve wireless channel conditions for multi-beam forming~\cite{MA_MultiBeam}. Specifically, as indicated in~\cite{MA_MultiBeam}, compared to fixed-position antenna (FPA) arrays, MA can achieve the maximum beamforming gain at different directions, which inspires us to utilize MAs in SR systems for achieving satisfactory primary and secondary transmission performance simultaneously.

	In this letter, we propose an MA enabled SR system to establish the mutually beneficial coexistence of both primary and secondary transmissions. Specifically, a primary transmitter (PT) equipped with a linear MA array is deployed to transmit primary information to the primary user (PU) with one single fixed-position antenna. The BD, with one single fixed-position antenna, delivers its own information to the PU by piggybacking on the primary transmission. Under this setup, the positions of MAs are optimized at the PT to achieve the maximum rates for both primary and secondary transmissions. Furthermore, we analyze the beam gains for both MA and FPA scenarios and confirm the superiority of the MA scheme. Numerical results verify the effectiveness of our proposed MA scheme in achieving mutualism.
	%Additionally, one optimal position of MAs is provided. 

	\section{System Model}
	
	As illustrated in Fig. \ref{System_model}, in the proposed MA enabled SR system, there exists a PT with $N$ MAs, a single-antenna~BD, and a single-antenna PU. We consider the CSR setup. Specifically, the PT delivers primary signals to the PU with the involvement of the BD. At the same time, the BD conveys its secondary information to the PU via overlaying onto the primary RF signals~\cite{SR_BD}. The positions of antennas at both the BD and the PU are fixed, while those of the PT are movable within a linear array. We consider a two-dimensional coordinate system, with the linear array positioned along the x-axis. Without loss of generality, the positions of MAs can be represented as $\bm x = [x_1,\ldots,x_n,\ldots,x_N]^T$, where $1\le n \le N$.	 
   \subsection{Channel Model} 
	
	\begin{figure}
		\centering
		\includegraphics[width=0.7\linewidth]{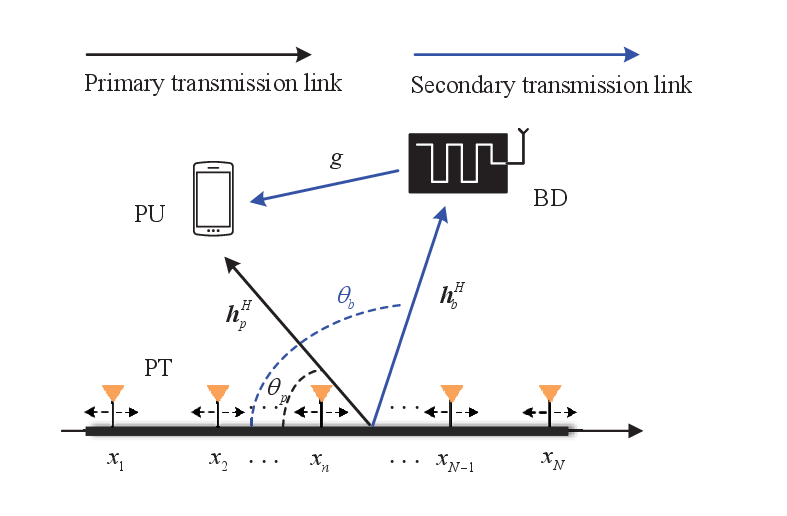}
		\caption{The linear MA array enabled SR systems.}
		\label{System_model}
	\end{figure}
	
	We consider the  quasi-static free-space line of sight (LoS) channel model for low-mobility scenarios~\cite{farfield_respose}, based on which the channel state information (CSI) remains unchanged with the positions of MAs being fixed.\footnote{Taking the smart factory as an example. Since the positions of the PT, PU (i.e., robot), and BD are fixed before the robot starts working, the related CSI between them is mainly controlled by the LoS link, which is thus unchanged. Under this setup, we can adjust the locations of MAs in accordance with the specific robot for designing the transmit beamforming vector.} We respectively denote the array response vectors of the PT with respect to (w.r.t.) $\theta_p$  and $\theta_b$ as
	\begin{align}
	 \bm \alpha_p \left(\bm x,\theta_p \right) = \left[ e^{-j\frac{2\pi}{\lambda} x_1 \cos{\theta_p}},\ldots,e^{-j\frac{2\pi}{\lambda} x_N \cos{\theta_p}}\right] ^T,
	  	\end{align}
  	\begin{align}
	\bm \alpha_b \left(\bm x,\theta_b \right) = \left[ e^{-j\frac{2\pi}{\lambda} x_1 \cos{\theta_b}},\ldots,e^{-j\frac{2\pi}{\lambda} x_N \cos{\theta_b}}\right] ^T,
	\end{align}
	where $\theta_p \in \left(0, \pi \right) $ and  $\theta_b \in \left(0, \pi \right)$ are  the steering angles of the PU and the BD w.r.t. the PT, respectively.
	Consequently, the channel vectors from the PT to the PU and  BD can be respectively represented as
	\begin{align}\label{PU_channel}
			\bm h_p^H& =  {\beta_p } e^{-j \frac{2\pi}{\lambda} d_p} \bm \alpha_ p^H \left(\bm x,\theta_p \right), \\
		    \bm h_b^H& =  {\beta_b  }  e^{-j \frac{2\pi}{\lambda} d_b}\bm \alpha_b^H \left(\bm x,\theta_b \right), 
	\end{align}
	where $ \beta_p e^{-j \frac{2\pi}{\lambda} d_p}$ and $ \beta_b e^{-j \frac{2\pi}{\lambda} d_b}$ respectively represent the complex path gains from the PT to the PU and the BD with $\beta_p=\frac{\lambda}{4\pi d_p} $ and $\beta_b=\frac{\lambda}{4\pi d_b} $, $\lambda$ represents the carrier wavelength, $d_p$ and $d_b$ are the distances from the PU and the BD to the origin, respectively.\footnote{{In this paper, we consider the far-field channel model as the aperture of MAs is much smaller than the distance for signal transmission~\cite{MA_Capacity,MA_MultiBeam,Secure_MA}, while the more accurate near-field channel model~\cite{farfield_respose} will be considered in our future work.} According to this model, both the angle between each MA and the PU/BD and the path gain remain constant.} 
	Similarly, the channel gain between the BD and the PU is expressed as %$g = \frac{\lambda}{4\pi d_{s}} e^{-j \frac{2\pi}{\lambda} d_s}$,
	\begin{align}
	     g = \frac{\lambda}{4\pi d_{s}} e^{-j \frac{2\pi}{\lambda} d_s}, 
	\end{align}
	where $d_s$ represents the related distance. Different from the FPA scenario, the spatial angle of the array response vector at the PT can be flexibly manipulated in the MA scenario by adjusting the positions of antennas for achieving higher beam gain.
	
	%the array response vector the PT i.e., $\bm \alpha\left( \bm x, \theta\right) $, in the MA scenario can be manipulated by adjusting the positions of antennas for achieving higher beam gain at specific angle $\theta$.

	\subsection{Transmission Model}

	Denote the primary transmission signal and secondary transmission signal as $s\left( l\right) \sim \mathcal{CN}\left( 0,1\right)  $ and $c \sim \mathcal{CN}\left( 0,1\right) $, respectively. As the CSR setup is adopted, during each secondary symbol period, $L$ primary symbols can be transmitted by the PT, where $L \gg 1$. Accordingly, the $l$-th received signal at the PU can be expressed as
	\begin{align}\label{received_signal}
		y\left(l \right)  =\bm h_p^H \bm w s\left( l\right)  + g \bm h_b^H \bm w c s\left( l\right) + z(l),
	\end{align} 
	where $\bm w$ is the transmit beamforming vector at the PT, and its maximum allowable transmit power is $P_t$, i.e., $\left|\left|  \bm w \right|\right|  ^2 \le P_t$.
	$z(l) \sim \mathcal{CN}\left( 0,\sigma^2\right) $ is the additive white Gaussian noise (AWGN). At the PU, the primary signal $s\left( l\right)$ is first decoded, followed by the decoding of the secondary signal $c$.
	 Consequently, the signal to noise ratio (SNR) at the PU for decoding  $s\left( l\right)$, denoted as $\gamma_p(c)$, is expressed as $ \frac{|\bm h_p^H \bm w +c g \bm h_b^H \bm w  |^2}{\sigma^2}$. Then, the average SNR at the PU for decoding $s\left( l\right)$ is obtained as
	\begin{align}\label{Average_SNR}
		\overline{\gamma}_p &= \mathbb{E}_c \left[\frac{|\bm h_p^H \bm w +c g \bm h_b^H \bm w  |^2}{\sigma^2} \right]  \nonumber \\
		 &\overset{(a)}{=} \frac{|\bm h_p^H \bm w|^2 + | g \bm h_b^H \bm w  |^2}{\sigma^2}, 
	\end{align} 
	where $(a)$ holds due to the zero mean and unit variance of $c$. According to~\eqref{Average_SNR},  the corresponding primary rate is
%	\begin{align}\label{Average_achevable_rate}
%		 R_p& = \log_2  \left( 1+ \frac{   \overset{p_1}{\overbrace{|\bm h_p^H \bm w|^2 }}  + \overset{{p_2}}{\overbrace{| g \bm h_b^H \bm w  |^2}}}{\sigma^2} \right),
%	\end{align} 
%	where $p_1$ in~\eqref{Average_achevable_rate} represents the signal power of the primary transmission link and $p_2$ is that of the secondary transmission link. Due to the double-fading effect of the secondary transmission link, the channel gain of the secondary transmission link (i.e., $ \left|\left| g \bm h_b^H\right| \right|^2  $) is significantly lower than that of the primary transmission link (i.e., $ \left|\left|  \bm h_p^H\right| \right|^2  $). Specifically, we have $ \frac{\left|\left| g \bm h_b^H\right| \right|^2}{\left|\left|  \bm h_p^H\right| \right|^2 }  = \left( \frac{\lambda d_p}{4\pi d_s d_b}  \right) ^2 < 10^{-4}$ with the simulation parameters provided in~\eqref{section-IV}. Therefore, the fact that $p_1 \gg p_2$ holds, and the primary rate can be approximated as
\begin{align}\label{Average_achevable_rate}
	R_p& = \log_2  \left( 1+ \frac{  {|\bm h_p^H \bm w|^2 }  +{| g \bm h_b^H \bm w  |^2}}{\sigma^2} \right),
\end{align} 
where $|\bm h_p^H \bm w|^2$ in~\eqref{Average_achevable_rate} represents the signal power of the primary transmission link and $| g \bm h_b^H \bm w  |^2$ is that of the secondary transmission link. Due to the double-fading effect of the secondary transmission link, the channel gain of the secondary transmission link (i.e., $ \left|\left| g \bm h_b^H\right| \right|^2  $) is significantly lower than that of the primary transmission link (i.e., $ \left|\left|  \bm h_p^H\right| \right|^2  $). Specifically, we have $ \frac{\left|\left| g \bm h_b^H\right| \right|^2}{\left|\left|  \bm h_p^H\right| \right|^2 }  = \left( \frac{\lambda d_p}{4\pi d_s d_b}  \right) ^2 < 10^{-4}$ with the simulation parameters provided in~\eqref{section-IV}. Therefore, the fact that $|\bm h_p^H \bm w|^2 \gg | g \bm h_b^H \bm w  |^2$ holds, and the primary rate can be approximated as
	\begin{align}\label{R_p}
		 R_p \approx \log_2 \left( 1+ { {{|\bm h_p^H \bm w|^2 }}} /{\sigma^2} \right). 
	\end{align}
	\begin{remark}\label{remark1}
		Although the existence of secondary transmission in~\eqref{Average_achevable_rate} can provide an additional link for the primary transmission, it does not consistently improve the primary rate. On the contrary, it may even lead to a deterioration in the primary transmission performance. This is because the design of transmit beamforming  needs to take into account the QoS requirements for both primary and secondary transmissions, i.e., there exists a  tradeoff   between the performance of primary and secondary transmissions~\cite{SR_BD}. In particular, achieving satisfactory secondary transmission performance may come at the expense of sacrificing the primary transmission performance. Furthermore, the deterioration in the primary transmission becomes pronounced when there is  low channel correlation between $\bm h_p^H$ and $\bm h_b^H$.
	\end{remark}
   In order to achieve the symbiosis in SR systems, i.e., both primary and secondary transmissions achieve their optimal performance, 
	 we propose to utilize the MAs at the PT to enhance the correlation between $\bm h_p^H$ and $\bm h_b^H$, aiming to ensure the optimal primary transmission performance and increase the secondary transmission performance. 
   	To accomplish this goal, the maximum ratio transmission (MRT) scheme is adopted for designing the transmit beamforming at the PT, which can be expressed as
	\begin{align}\label{MRT}
		\bm w \left( \bm x,\theta_p \right)    =\sqrt{P_t/N}  \bm \alpha_ p \left(\bm x,\theta_p \right).
	\end{align}
	\eqref{MRT} guarantees the optimal primary transmission performance as defined in~\eqref{R_p}. After obtaining $s\left( l\right)$, the successive interference cancellation (SIC) technique is employed at the PU to eliminate the interference from the primary transmission link, i.e., $\bm h_p^H \bm w s\left( l\right)$ in \eqref{received_signal}. Subsequently, the maximal ratio combining (MRC) technique is employed for decoding $c$, and the corresponding secondary rate is expressed as
	\begin{align}\label{SU_achevable_rate}
		R_c = \frac{1}{L} \log_2( 1 + \frac{L \left| g\right|^2  |  \bm h_b^H \bm w \left( \bm x,\theta_p \right)  |^2}{\sigma^2}),
	\end{align}
	{which is related to both the transmit beamforming at the PT and the positions of the MAs.  In the FPA scenario, with the designed transmit beamforming vector $\bm w \left( \bm x,\theta_p \right)$ in~\eqref{MRT}, the secondary rate remains consistently low as a constant.} However, in the MA scenario,   the positions of MAs (i.e., $\bm x$) can be optimized to  achieve mutually beneficial coexistence for the primary and secondary transmissions, hence providing an opportunity for symbiosis.

%   {\color{blue}\begin{theorem}\label{Theorem1}
%   		The maximization of the secondary rate $R_c$ with the designed transmit beamforming vector $\bm w \left( \bm x,\theta_p \right)$ in \eqref{MRT} does not affect the optimality of the primary rate $R_p$.
%   \end{theorem}
%	\begin{proof}\label{proof}
%		By substituting $\bm w \left(\bm x, \theta_p \right) $ into the first term of equation \eqref{Average_achevable_rate}, we obtain
%		\begin{align}
%			   |\bm h_p^H \bm w \left(\bm x, \theta_p \right)|^2 &=\left|  \frac{1}{N}\sqrt{P_t}\beta_p \bm \alpha_ p^H \left(\bm x,\theta_p \right) \bm \alpha_ p \left(\bm x,\theta_p \right)\right| ^2  \nonumber  \\  
%			   &=  \left( \sqrt{P_t}\beta_p\right) ^2,
%		\end{align}
%		which is a constant, and the second term have the same form as shown in equation \eqref{SU_achevable_rate}. 
%		Thus, \textbf{Theorem}~\ref{Theorem1} holds true. 
%	\end{proof}
%	\begin{remark}
%		\textbf{Theorem}~\ref{Theorem1} demonstrates that the designed transmit beamforming vector $\bm w \left( \bm x,\theta_p \right)$ enables the primary and secondary transmissions to achieve mutually beneficial coexistence in the MA scenario, aligning with the concept of symbiosis. 
%	\end{remark}
%	In the FPA scenario, with the designed $\bm w \left( \bm x,\theta_p \right)$, the secondary  rate  remains consistently low as a constant. However, in the MA scenario, variations in the positions of MAs (i.e., $\bm x$) provide opportunities for enhancing both primary and secondary  rates, ultimately leading to mutualistic outcomes.}

	\section{Problem Formulation and  Proposed Solution}\label{section3}
	
	{In this section, our objective is to maximize the secondary rate defined in~\eqref{SU_achevable_rate} through optimizing the transmit beamforming vector and the positions of the MAs. As $R_c$ increases monotonically with $  \left| g \right|^2 \left| \bm h_b^H \bm w \left( \bm x,\theta_p \right) \right| ^2 $, incorporating~\eqref{MRT} into~\eqref{SU_achevable_rate}, the goal of maximizing the secondary rate is equivalent to maximizing the beam gain at angle $\theta_b$, i.e.,  $ \left| \bm \alpha_b^H \left(\bm x,\theta_b \right)  \bm \alpha_p \left(\bm x,\theta_p \right) \right|  $, which is solely related to $\bm x$.} Consequently, we formulate the optimization problem as follows
	 \begin{equation}\tag{$\textbf{P1}$} 
	 	\begin{aligned}
	 		\max_{\bm{x}}  &~~ \kappa \left(\bm x, \theta_p, \theta_b \right)   \\ 
	 		\text{s.t.}~~&|x_i - x_j| \ge d_{\text{min}},~ 1\le i \neq j \le N,
	 	\end{aligned}
	 \end{equation}
     {where  $\kappa \left(\bm x, \theta_p, \theta_b \right) =  \left| \bm \alpha_b^H \left(\bm x,\theta_b \right)  \bm \alpha_p \left(\bm x,\theta_p \right) \right| $ represents the beam gain of the PT at angle $\theta_b$ with the designed transmit beamforming vector $\bm w \left(\bm x,\theta_p \right)$,}  $d_{\text{min}} = \frac{\lambda}{2}$ indicates the minimum distance between arbitrary two MAs for avoiding the coupled effect~\cite{MA_mag,MA_Capacity,Secure_MA,MA_MultiBeam}.  
    The non-convex constraint makes it challenging to solve (\textbf{P1}). Before addressing this challenge, we first analyze the objective function of (\textbf{P1}) (i.e., beam gain at angle $\theta_b$) for the MA scenario, and subsequently present the corresponding optimal position. Furthermore, we delve into the objective function of (\textbf{P1}) for the FPA scenario to show the superiority of the proposed MA scheme.
     
	\subsection{Beam gain in the MA scenario}\label{Section3-A}
%	Considering the limitations of the FPA scenario, we propose an MA scheme to overcome this issue and achieve the goal of symbiosis. 
%	For the MA scenario, the representation of beam gain at angle $\theta_b$ can be expanded as follows
	In the MA scenario, for two distinct angles $\theta_p$ and $\theta_b$,\footnote{We omit the case $ \cos{\theta_p}=\cos{\theta_b} $, in which the steering vector of $\theta_p$ is identical to that of $\theta_b$.} we expand the expression of $  \kappa\left(\bm x,\theta_p,\theta_b  \right) $ as follows
	\begin{align}\label{MA_beamgain}
		&\kappa\left(\bm x,\theta_p,\theta_b  \right) 
		=\left| \bm \alpha_b^H \left(\bm x,\theta_b \right)  \bm \alpha_p \left(\bm x,\theta_p \right) \right|  \nonumber  \\
		& =\left| \bm \alpha_p^H \left(\bm x,\theta_p \right)  \bm \alpha_b \left(\bm x,\theta_b \right) \right|  
		= \left|\sum_{n=1}^{N} e^{j\frac{2 \pi}{\lambda} x_n \Delta\left(\theta_p,\theta_b \right)  } \right|,
		%	&=\left| \sum_{n=1}^{N} e^{-j\frac{2\pi}{\lambda} x_n \cos{\theta_b} }  e^{j\frac{2\pi}{\lambda} x_n \cos{\theta_p}}  \right| 
		%	= \left|\sum_{n=1}^{N} e^{j n\pi \Delta\left(\theta_p,\theta_b \right)  } \right|  \nonumber 
	\end{align}
%	\begin{align}
%		\kappa \left( \bm x,\theta_p, \theta_b  \right)  =  \left|\sum_{n=1}^{N} e^{j\frac{2 \pi}{\lambda} x_n \Delta\left(\theta_p,\theta_b \right)  } \right|.
%	\end{align}
	where  $ \Delta\left(\theta_p,\theta_b \right)  = \cos{\theta_p}-\cos{\theta_b} $. As $\kappa\left(\bm x,\theta_p,\theta_b  \right)$ is an even function w.r.t. $  \Delta\left(\theta_p,\theta_b \right) $, it is reasonable to investigate the case that $  \Delta\left(\theta_p,\theta_b \right) >0$, i.e., $\cos{\theta_p}>\cos{\theta_b} $. As the phase difference between $\theta_p$ and $\theta_b$ increases, $\Delta\left(\theta_p,\theta_b \right)$ gradually approaches 2. Therefore, we have $  \Delta\left(\theta_p,\theta_b \right)  \in (0,2) $. To obtain the maximum value of $\kappa\left(\bm x,\theta_p,\theta_b  \right) $, the optimal positions of MAs are presented in \textbf{Theorem~\ref{Theorem3}}.
	%In the FPA scenario, the beam gain remains constant at fixed angles $\theta_p$ and $\theta_b$ with the pre-designed transmit beamforming vector. However, in the MA scenario, to achieve the maximum value of  $ \kappa \left( \bm x,\theta_p, \theta_b  \right) $, the positions of MAs (i.e., $ \bm x  $) should satisfy
%	\begin{align}\label{optimal_condiation}
%		\left( x_{n+1} -x_n\right)  \frac{\Delta\left(\theta_p,\theta_b \right)}{\lambda} = k,~ 1\le n \le N-1,
%	\end{align}
%	where $k$ is positive integer. Thanks to the movable characteristic of MAs, the maximum beam gain at angle $\theta_b$ can be attained. 
	\begin{theorem}\label{Theorem3}
		For the given $ \Delta\left(\theta_p,\theta_b \right) $, one optimal solution to problem (\textbf{P1}) is
		\begin{align}\label{optimal_MA}
			x_n =\left(  n-1 \right) \frac{\lambda}{ \Delta\left(\theta_p,\theta_b \right)},~1\le n \le N.
		\end{align}
	\end{theorem}
	\begin{proof}
		With the solution provided in \eqref{optimal_MA}, we can express $x_{n+1} =x_n + \frac{\lambda}{ \Delta\left(\theta_p,\theta_b \right)}$, where $1\le n \le N$. Accordingly, the expression for $x_n$ can be reformulated as	
		\begin{align}\label{x_n}
			x_n = x_1 + (n-1)\frac{\lambda}{ \Delta\left(\theta_p,\theta_b \right)}, ~1\le n \le N.
		\end{align} 
		Substituting~\eqref{x_n} into~\eqref{MA_beamgain}, $\kappa\left(\bm x,\theta_p,\theta_b  \right) $ can be computed by
		\begin{align}\label{Eq15}
			\kappa\left(\bm x,\theta_p,\theta_b  \right)  &= \left|\sum_{n=1}^{N} e^{j\frac{2 \pi \Delta\left(\theta_p,\theta_b \right) }{\lambda}\left(   x_1 + (n-1)\frac{\lambda}{ \Delta\left(\theta_p,\theta_b \right)}\right)   } \right|  \nonumber\\
			& =\left| N e^{j\frac{2 \pi \Delta\left(\theta_p,\theta_b \right) }{\lambda}  x_1    } \right| = N.
		\end{align} 
	Without loss of generality, let $ x_1 $ be set to 0, thus completing the proof.
	 %Therefore, \textbf{Theorem~\ref{Theorem3}} holds true.
	\end{proof}
	\begin{remark}
		In \textbf{Theorem~\ref{Theorem3}}, one optimal solution to \textbf{P1} is presented, and there also exist  other optimal solutions. It is worth noting that the solution provided in \textbf{Theorem~\ref{Theorem3}} complies with the constraint in \textbf{P1} due to $  \Delta\left(\theta_p,\theta_b \right)  \in \left( 0,2\right) $. Furthermore, \textbf{Theorem~\ref{Theorem3}} indicates that we can simultaneously achieve the maximum primary and secondary  rates in the MA scenario, aligning with the concept of symbiosis. 
		This is due to the fact that the presented optimal positions of MAs ensure the maximum beam gain at angle $\theta_b$, {while the designed transmit beamforming vector guarantees the maximum beam gain at angle $\theta_p$.} 
		Notably, when $ \Delta\left(\theta_p,\theta_b \right) $ is small, the distance between adjacent antennas is much larger than $\lambda$, which results in a larger size of linear array.
	\end{remark}

	\subsection{Beam gain in the FPA scenario}\label{Section3-B}
	
		Different from the MA scenario, the distance between adjacent antennas is fixed in the FPA scenario. Thus, we have $x_{n+1} - x_{n} = d_\text{min} =\frac{\lambda}{2}$, where $1\le n \le N-1$. To simplify the analysis process, we make the assumption that $x_1 =\frac{\lambda}{2}$, and this assumption does not affect the analysis. Accordingly, the positions of FPA can be represented as $\bm x_F = [\frac{\lambda}{2},\ldots,n\frac{\lambda}{2},\ldots,N\frac{\lambda}{2}]^T$. Substituting $\bm x_F$ into~\eqref{MA_beamgain}, the beam gain at angle $\theta_b$ can be accordingly expressed as
		%For two distinct angles $\theta_p$ and $\theta_b$,\footnote{We omit the case $ \cos{\theta_p}=\cos{\theta_b} $, in which the steering vector of $\theta_p$ is identical to that of $\theta_b$.} we expand the expression of $  \kappa\left(\bm x_F,\theta_p,\theta_b  \right) $ as follows
	\begin{align}\label{Eq16}
	&\kappa\left(\bm x_F,\theta_p,\theta_b  \right) =  \left|\sum_{n=1}^{N} e^{j n\pi \Delta\left(\theta_p,\theta_b \right)  } \right|.
%	&=\left| \sum_{n=1}^{N} e^{-j\frac{2\pi}{\lambda} x_n \cos{\theta_b} }  e^{j\frac{2\pi}{\lambda} x_n \cos{\theta_p}}  \right| 
%	= \left|\sum_{n=1}^{N} e^{j n\pi \Delta\left(\theta_p,\theta_b \right)  } \right|  \nonumber 
	\end{align}
	\begin{theorem}\label{Theorem2}
		The beam gain in the FPA scenario satisfies $\kappa\left(\bm x_F,\theta_p,\theta_b  \right) < N$ when $  \Delta\left(\theta_p,\theta_b \right)  \in \left( 0,2\right) $ and $\theta_p \neq \theta_b$.
	\end{theorem}
	\begin{proof}	
		It can be readily observed that $ \left|\sum_{n=1}^{N} e^{j n\pi \Delta\left(\theta_p,\theta_b \right)  } \right|  \overset{(\varsigma)}{\le} \sum_{n=1}^{N}\left| e^{j n\pi \Delta\left(\theta_p,\theta_b \right)  } \right| = N $, and the equality holds if and only if $ \Delta\left(\theta_p,\theta_b \right)  =0$ or $2 $.  Given that  $  \Delta\left(\theta_p,\theta_b \right) \in (0,2)$, we have $\kappa (\bm x_F ,\theta_p ,\theta_b)<N$.
	%	See Appendix \eqref{ProofofTheorem2}.
	\end{proof}
	\begin{remark}
		\textbf{Theorem~\ref{Theorem2}} implies that, {in the FPA scenario, employing the designed transmit beamforming vector to maximize the beam gain at angle $\theta_p$ is not feasible to simultaneously achieve the highest beam gain at angle $\theta_b$.} In essence, ensuring the maximization of primary rate generally results in a relatively low secondary rate, which will be discussed in Section~\ref{section-IV}. While the proposed MA scheme can achieve both the optimal primary and secondary performance, which provides an opportunity for mutualism.
	\end{remark}

	\begin{figure*}[t]
		\centering
		\subfigure[$\Delta\left(\theta_p,\theta_b \right)=0.5$.]{\label{Delta1}
			\includegraphics[height=3.6cm]{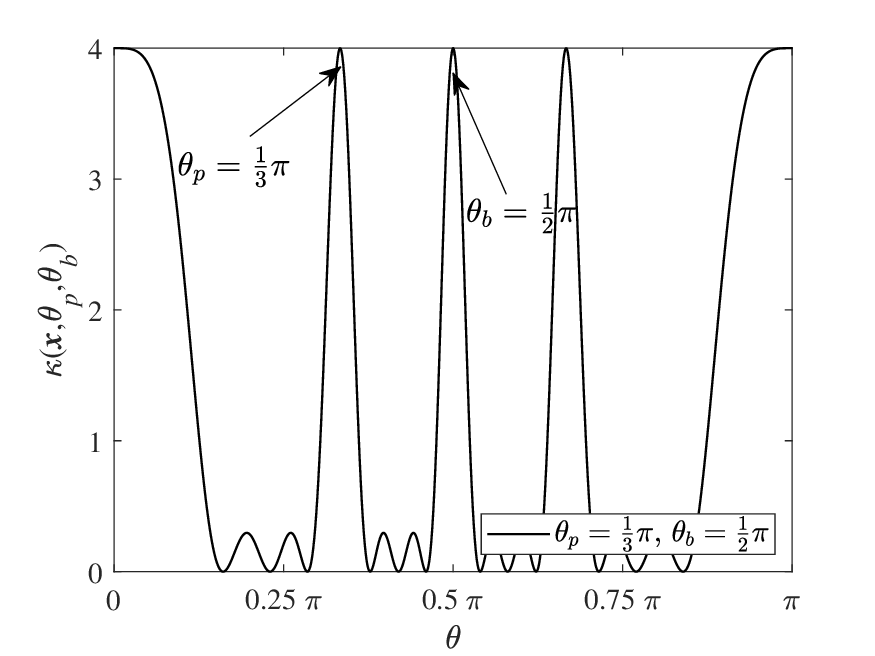}}
		\hspace{5pt}
		\subfigure[$\Delta\left(\theta_p,\theta_b \right)=1.2071$.]{\label{Delta2}
			\includegraphics[height=3.6cm]{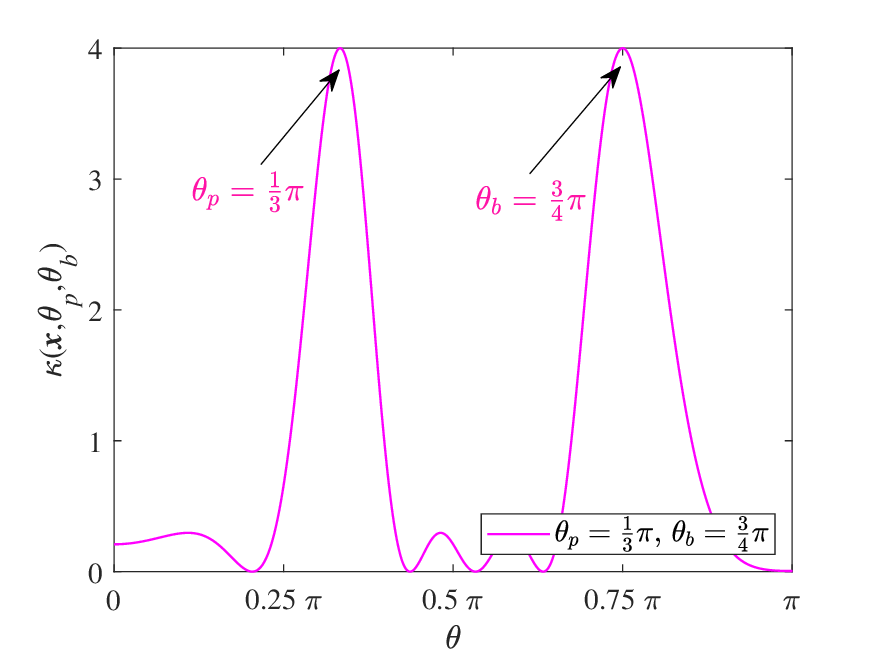}}
		\hspace{5pt}
		\subfigure[$\Delta\left(\theta_p,\theta_b \right)=1.4397$.]{\label{Delta3}
			\includegraphics[height=3.6cm]{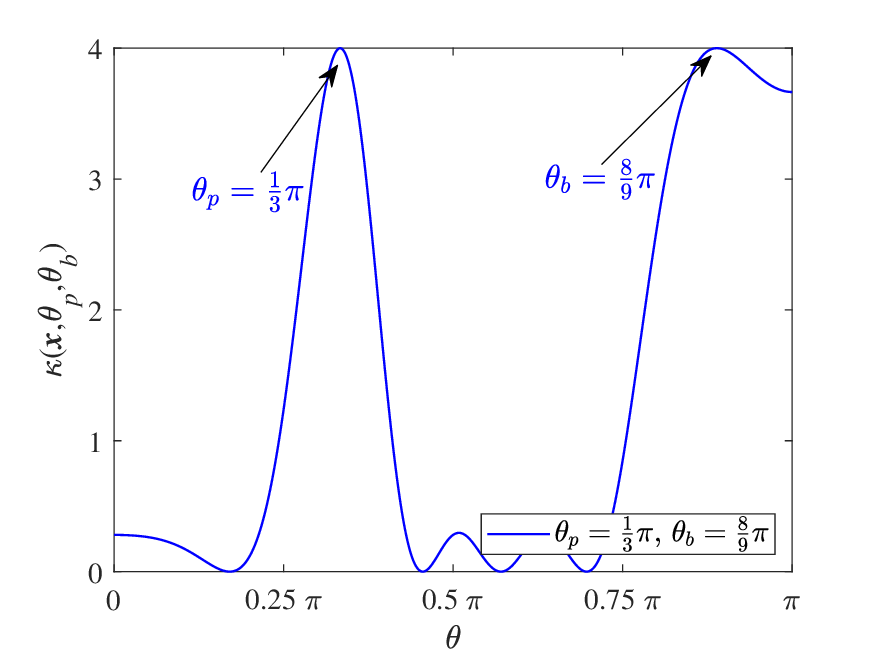}}
		\caption{Beam gain of the MA scheme with different $\Delta\left(\theta_p,\theta_b \right)$ in Table \ref{Tab1}.}
		\label{Delta}
		\vspace{-5pt}
	\end{figure*} 
	
	\section{Numerical Results}\label{section-IV}

	In this section, we present the simulation results to demonstrate the superiority of the proposed MA enabled SR system. {Specifically, we consider a two-dimensional simulation scenario, and the parameters are set as Table~\ref{SimulationParameters}.}

	\begin{table}[htb]\scriptsize
		\begin{center}{
		\caption{Simulation Parameters}
		\label{SimulationParameters}
		\begin{tabular}{c|c}
		\hline 
		\textbf{Parameter} & \textbf{Value}          
		\\ \hline Number of MAs, $N$ & $4$
		\\ \hline Distance from the PT to the PU, $d_p$, & $40$ m
		\\ \hline Distance from the PT to the BD, $d_b$, & $20 \sqrt{3}$ m
		\\ \hline Position of PU  & $\left(-d_p \cos\left(\theta_p \right), d_p \sin\left(\theta_p \right)\right)$  
		\\ \hline Position of BD  & $\left(-d_b \cos\left(\theta_b \right), d_b \sin\left(\theta_b \right)\right) $ 
		\\ \hline Carrier wavelength, $\lambda$ & $0.5$ m
		\\ \hline Spread factor, $L$ & $15$
		\\ \hline Noise power, $\sigma^2$ & $-80$ dBm
		\\ \hline
		\end{tabular}}
	\end{center}
\end{table}

\begin{table}[h]\scriptsize
	\centering  
	\renewcommand\arraystretch{2}
	\caption{The optimal positions of the MAs.}\label{Tab1}
	\resizebox{\linewidth}{!}{
		\begin{tabular}{|c|c|c|c|c|c|}
			\hline
			& $\Delta\left(\theta_p,\theta_b \right)$                        & $x_{1}$ & $x_{2}$         & $x_{3}$             & $x_{4}$           \\ \hline
			$\theta_p=\frac{1}{3}\pi$ and $\theta_b=\frac{1}{2}\pi$& 0.5     & 0       & 2$\lambda$       & 4$\lambda$         & 6$\lambda$        \\ \hline
			$\theta_p=\frac{1}{3}\pi$ and $\theta_b=\frac{3}{4}\pi$& 1.2071	 & 0       & 0.8284$\lambda$  & 1.6569$\lambda$    & 2.4853 $\lambda$  \\ \hline
			$\theta_p=\frac{1}{3}\pi$ and $\theta_b=\frac{8}{9}\pi$& 1.4397  & 0       &0.6946 $\lambda$  & 1.3892$\lambda$    &2.0838$\lambda$    \\ \hline
	\end{tabular}}
\end{table}

In Table~\ref{Tab1}, according to \textbf{Theorem~\ref{Theorem3}}, we present the optimal positions of MAs for various $\theta_p$ and $\theta_b$  with $N=4$.  As observed in Table~\ref{Tab1}, as the decrease of $\Delta\left(\theta_p,\theta_b \right)$ (i.e., the phase difference between  $\theta_p$ and $\theta_b$), the aperture of liner array becomes large. Even if $\Delta\left(\theta_p,\theta_b \right)$ is small, we can still reach both the maximum primary and secondary  rates, but at the cost of deploying  a linear array with a sufficiently large aperture.  
To validate the optimality of the positions of MAs as listed in Table~\ref{Tab1},  we present in Fig.~\ref{Delta} the beam gains for various $\Delta\left(\theta_p,\theta_b \right)$ using the proposed transmit beamforming scheme in~\eqref{MRT} and the optimal MA positions derived from \textbf{Theorem~\ref{Theorem3}}. As depicted in Fig.~\ref{Delta}, employing the proposed transmit beamforming scheme consistently results in the maximum beam gain at the PU (i.e., $\theta_p= \frac{1}{3}\pi$), ensuring attainment of the peak primary  rate. 
Furthermore, the proposed MA scheme can also achieve the upper bound beam gain at different angles $\theta_b$, thereby confirming the effectiveness of position optimization. The numerical results illustrated in Fig.~\ref{Delta} demonstrate that MA enabled SR can achieve maximal beam gains at both the PU and the BD, underscoring its potential to facilitate mutualism in SR systems.

 \begin{figure}[t]
 	\centering
 	\subfigure[$\theta_p=\frac{1}{2}\pi $, $\theta_p=\frac{2}{3}\pi$ and $N=4$.]{\label{Nt4BeamGain}
 		\includegraphics[height=3.08cm]{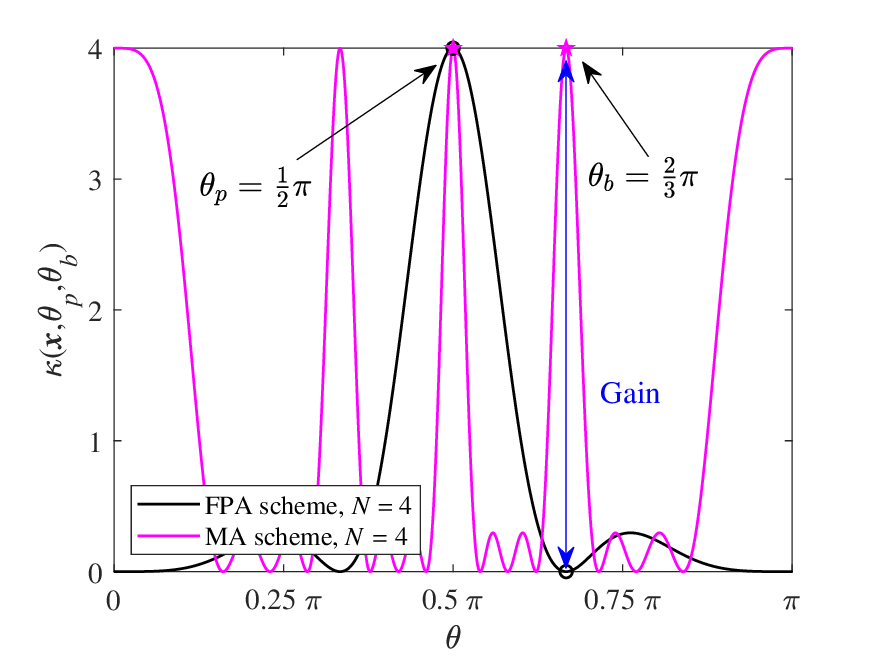}}
 	\hspace{5pt}
 	\subfigure[$\theta_p=\frac{1}{4}\pi $, $\theta_p=\frac{3}{4}\pi$ and $N=6$.]{\label{Nt6BeamGain}
 		\includegraphics[height=3.08cm]{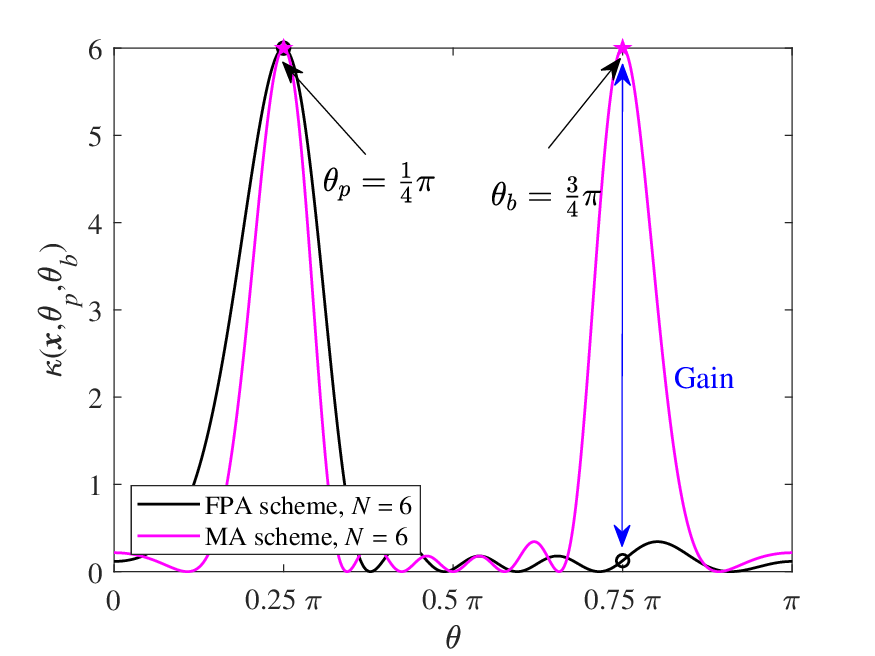}}
 	\caption{Beam gain with different $\Delta\left(\theta_p,\theta_b \right)$ and $N$.}
 	\label{BeamGaincomparison}
 	\vspace{-5pt}
 \end{figure}

	To illustrate the superiority of the proposed scheme, in Fig.~\ref{BeamGaincomparison}, we conducted a comparison of beam gains between the MA scheme and the FPA scheme for $N = 4$ and $N = 6$, respectively. In Fig.~\ref{Nt4BeamGain}, the beam gain with $\theta_p=\frac{1}{2}$, $\theta_b=\frac{2}{3}\pi$, and $N = 4$ is presented. It is evident that the proposed MA scheme achieves maximum beam gains in both the PU direction and BD direction, whereas the FPA scheme solely ensures maximum beam gain in the PU direction, with a beam gain of 0 in the BD direction. Therefore, the proposed MA scheme effectively enhances the performance of secondary transmission. Furthermore, the beam gain of the FPA scheme is consistently observed to be very low in most directions. For example, as depicted in Fig.~\ref{Nt4BeamGain}, the beam gain is less than 2 when $0<\theta<0.4\pi$ and $0.6\pi<\theta<\pi$, and similarly in Fig.~\ref{Nt6BeamGain}, it is lower than 2 when $0<\theta<0.12\pi$ and $0.34\pi<\theta<\pi$. This indicates that, in most cases, the performance of secondary transmission is seriously limited in the FPA scenario. Conversely, in the MA scenario, optimizing the positions of MA  can effectively enhance the beam gain and subsequently improve secondary transmission performance.
	%Moreover, from Fig.~\ref{Nt4BeamGain} and Fig.~\ref{Nt6BeamGain}, it is found that for the FPA scenario, increasing the number of antennas may not lead to a higher beam gain at angle $\theta_b$ 
	%due to an increased likelihood of encountering a beam gain of 0. However, for the MA scenario, the maximum beam gain increases with the number of antennas.

%	\begin{figure}
%		\centering
%		\includegraphics[width=0.7\linewidth]{Nt4.eps}
%		\caption{Beam gain with $\theta_p=\frac{1}{2}\pi $, $\theta_p=\frac{2}{3}\pi$ and $N=4$.}
%		\label{Nt4BeamGain}
%	\end{figure} 
%	
%	
%	\begin{figure}
%		\centering
%		\includegraphics[width=0.7\linewidth]{Nt6.eps}
%		\caption{Beam gain with $\theta_p=\frac{1}{4}\pi $, $\theta_p=\frac{3}{4}\pi$ and $N=6$.}
%		\label{Nt6BeamGain}
%	\end{figure}

	\begin{figure}
		\centering
		\includegraphics[width=0.6\linewidth]{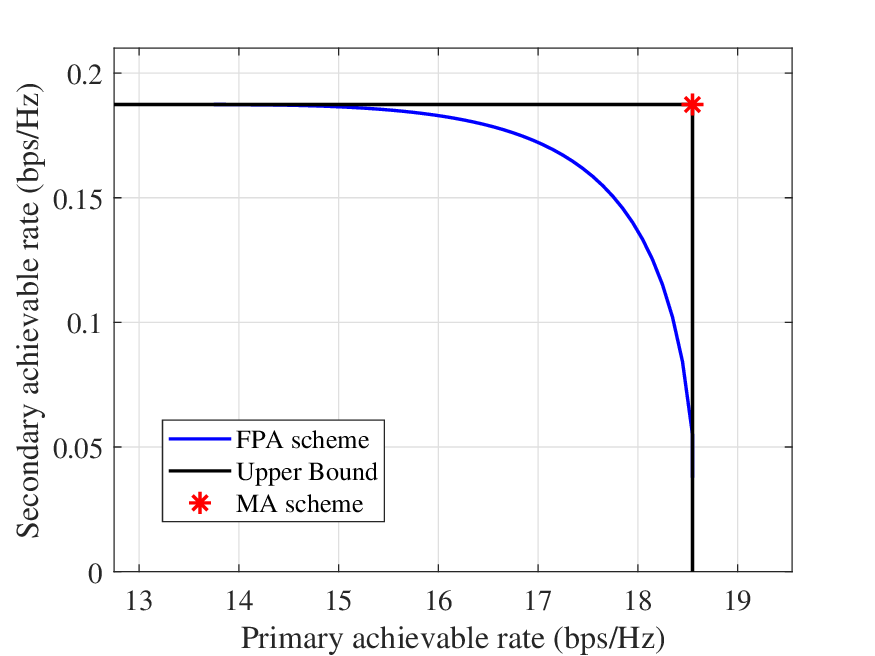}
		\caption{{The secondary achievable rate versus  
			the primary achievable rate with $\theta_p=\frac{1}{3}\pi$, $\theta_b=\frac{3}{4}\pi$ and $N=4$.}}
		\label{Rate-Region}
	\end{figure}
	\begin{figure}
		\centering
		\includegraphics[width=0.6\linewidth]{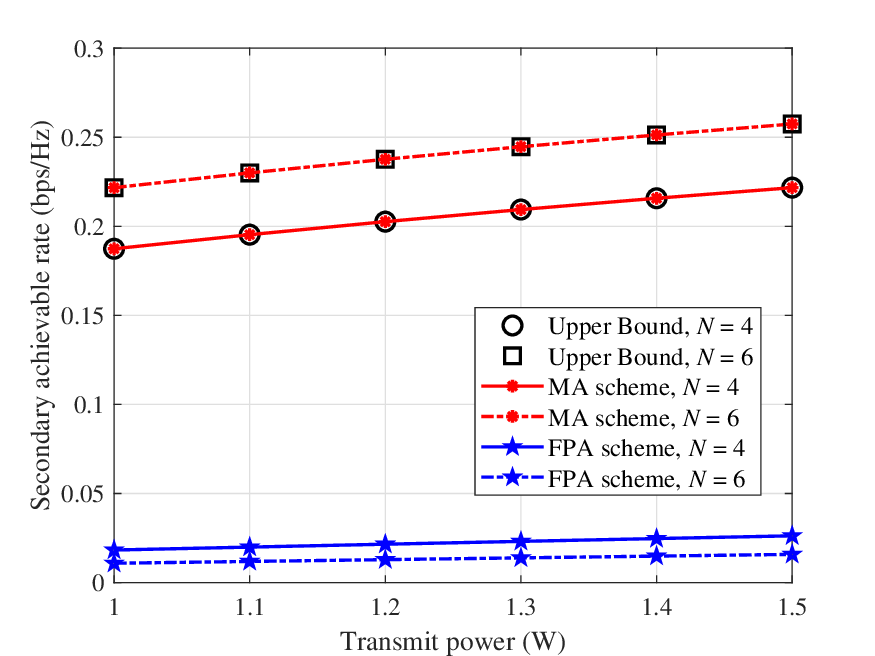}
		\caption{Secondary achievable rate versus the transmit power with $\theta_p=\frac{1}{3}\pi $ and $\theta_p=\frac{3}{4}\pi$.}
	\label{SR}
	\end{figure} 

	{In Fig.~\ref{Rate-Region}, the achievable rate region for primary and secondary transmissions is examined. The black solid line represents the maximum values (i.e., upper bounds) that the secondary and/or primary achievable rates can attain with the maximum allowable transmit power.
    It is evident that, with the FPA scheme, there exists a tradeoff between the performance of primary and secondary transmissions. Specifically, achieving satisfactory performance in secondary transmission may come at the cost of sacrificing primary transmission performance, and vice versa. Conversely, the proposed MA scheme can achieve optimal achievable rates for both primary and secondary transmissions, as indicated by the red asterisk in Fig.~\ref{Rate-Region}.}

	We investigate the secondary  rate versus the transmit power in Fig.~\ref{SR} with $\theta_p=\frac{1}{3}\pi$ and $\theta_b=\frac{3}{4}\pi$. Based on the results shown in Fig.~\ref{SR}, it can be observed that the secondary  rate increases for all schemes as the transmit power grows, indicating potential avenues for enhancing secondary transmission performance.  However, due to the low beam gain at the BD, increasing the transmit power in the FPA scenario yields only  marginal performance improvement for secondary transmission. The proposed MA scheme demonstrates the capability to achieve the upper bound of secondary transmission, surpassing the performance of the FPA scheme by approximately 11.4 times when $ N=6 $. This significant enhancement in secondary rate highlights the superiority of the proposed MA scheme.  Furthermore, it is observed that the secondary  rate can be enhanced by augmenting the number of antennas in the MA scheme, but it may be ineffective in the FPA scheme. {This phenomenon is attributed to the fact that the beam gain at angle $\theta_b$ in the FPA scheme, as defined by~\eqref{Eq16}, is not only dependent on the number of transmit antennas but also on $\Delta\left( \theta_p, \theta_b\right)$. However, in the MA scheme, the maximum beam gain increases with the number of antennas, as given by~\eqref{Eq15}.} 

   \section{Conclusion}

   In this letter, we proposed a novel MA enabled SR system, which holds promise for promoting mutualism of primary and secondary transmissions. Specifically, in this system,  the MRT beamforming  was adopted to ensure the optimal primary transmission performance. Based on this, the optimal positions of MAs were obtained with low complexity to achieve the highest secondary rate. Then, we proved the superiority of the proposed MA scheme over the FPA scheme through theoretical analysis, which was further validated by numerical results.

   %while controlling the positions of MAs to achieve an optimal secondary transmission performance, thereby guaranteeing both the highest primary and secondary rates. In contrast to the FPA scheme, the optimal primary and secondary transmission performance can always be achieved by deriving the optimal positions of MAs with low complexity. Numerical results validated the effectiveness of position optimization and demonstrated the superiority of our proposed scheme.

   %\appendix


\begin{thebibliography}{}
\providecommand{\url}[1]{#1}
\csname url@samestyle\endcsname
\providecommand{\newblock}{\relax}
\providecommand{\bibinfo}[2]{#2}
\providecommand{\BIBentrySTDinterwordspacing}{\spaceskip=0pt\relax}
\providecommand{\BIBentryALTinterwordstretchfactor}{4}
\providecommand{\BIBentryALTinterwordspacing}{\spaceskip=\fontdimen2\font plus
\BIBentryALTinterwordstretchfactor\fontdimen3\font minus
  \fontdimen4\font\relax}
\providecommand{\BIBforeignlanguage}[2]{{%
\expandafter\ifx\csname l@#1\endcsname\relax
\typeout{** WARNING: IEEEtran.bst: No hyphenation pattern has been}%
\typeout{** loaded for the language `#1'. Using the pattern for}%
\typeout{** the default language instead.}%
\else
\language=\csname l@#1\endcsname
\fi
#2}}
\providecommand{\BIBdecl}{\relax}
\BIBdecl

\end{thebibliography}


\begin{thebibliography}{10}
	\bibliographystyle{IEEEtran}
	
	\bibitem{IoT_Num}
	D. C. Nguyen \emph{et al.}, ``6G Internet of things: A comprehensive survey," \emph{IEEE Internet Things J.}, vol. 9, no. 1, pp. 359-383, 1 Jan.1, 2022.
	
	\bibitem{SR_TCCN}
	Y. -C. Liang, Q. Zhang, E. G. Larsson, and G. Y. Li, ``Symbiotic radio: Cognitive backscattering communications for future wireless networks," \emph{IEEE Trans. Cogn. Commun. Netw.}, vol. 6, no. 4, pp. 1242-1255, Dec. 2020.
	
	\bibitem{SR_Li}
	X. Li \emph{et al.}, ``Physical-layer authentication for ambient backscatter-aided NOMA symbiotic systems," \emph{IEEE Trans. Commun.}, vol. 71, no. 4, pp. 2288-2303, April 2023.
		

	\bibitem{SR_BD}
	R. Long, Y. -C. Liang, H. Guo, G. Yang, and R. Zhang, ``Symbiotic radio: A new communication paradigm for passive internet of things," \emph{IEEE Internet  Things J.},  vol. 7, no. 2, pp. 1350-1363, Feb. 2020.
	
	
	\bibitem{SR_Zeng}
	J. Xu, Z. Dai, and Y. Zeng, ``MIMO symbiotic radio with massive backscatter devices: Asymptotic analysis and precoding optimization," \emph{IEEE Trans. Commun.}, vol. 71, no. 9, pp. 5487-5502, Sept. 2023.
	
	\bibitem{Multi-Acc_Liang}
	J. Wang, X. Ding, Q. Zhang, and Y. -C. Liang, ``Multiple access design for symbiotic radios: Facilitating massive IoT connections with cellular networks," \emph{IEEE Trans.  Wireless Commun.}, vol. 23, no. 1, pp. 201-216, Jan. 2024.
	
	\bibitem{PSR_Yang}
	G. Yang, T. Wei, and Y. -C. Liang, ``Joint hybrid and passive beamforming for millimeter wave symbiotic radio systems," \emph{IEEE Wireless Commun. lett.}, vol. 10, no. 10, pp. 2294-2298, Oct. 2021.
	
	\bibitem{PSR_Zeng}
	Z. Dai, R. Li, J. Xu, Y. Zeng, and S. Jin, ``Rate-region characterization and channel estimation for cell-free symbiotic radio Communications," \emph{IEEE Trans. Commun.}, vol. 71, no. 2, pp. 674-687, Feb. 2023.
		
	\bibitem{MA_mag}
	L. Zhu, W. Ma, and R. Zhang, ``Movable antennas for wireless communication: Opportunities and challenges,"  \emph{IEEE Commun. Mag.}, vol. 62, no. 6, pp. 114-120, June 2024.
		
	\bibitem{You_TransceiverDesign}
	C. You, \emph{et al.}, ``Next generation advanced transceiver technologies for 6G and beyond." arXiv:2403.16458, 2024.
	
	\bibitem{MA_Capacity}
	W. Ma, L. Zhu, and R. Zhang, ``MIMO capacity characterization for movable antenna systems," \emph{IEEE Trans. Wireless Commun.}, vol. 23, no. 4, pp. 3392-3407, April 2024.
	
	\bibitem{Secure_MA}
	G. Hu, Q. Wu, K. Xu, J. Si, and N. Al-Dhahir, ``Secure wireless communication via movable-antenna array," \emph{IEEE Signal Process. Lett.}, vol. 31, pp. 516-520, 2024.
	

	\bibitem{MA_MultiBeam}
	W. Ma, L. Zhu, and R. Zhang, ``Multi-beam forming with movable-antenna array," \emph{IEEE Commun. Lett.}, vol. 28, no. 3, pp. 697-701, March 2024.
	
{
	\bibitem{farfield_respose}
	Y. Liu, Z. Wang, J. Xu, C. Ouyang, X. Mu, and R. Schober, ``Near-field communications: A tutorial review," \emph{IEEE Open J. Commun. Soc.}, vol. 4, pp. 1999-2049, 2023.
	}
	
	
	
	%C. Feng and Y. Zeng, ``When does UPW model become invalid for XL-MIMO with directional array elements?," \emph{IEEE Commun. Lett.}, vol. 28, no. 2, pp. 422-426, Feb. 2024.
	


		
	\end{thebibliography}
\end{document}